\documentclass[centertags,12pt]{amsart}

\usepackage{amssymb}
\usepackage{amsmath}
\usepackage{amsthm}
\usepackage{color}

\makeatletter
\renewcommand{\subsection}{\@startsection
{subsection}{2}{0mm}{\baselineskip}{-0.25cm}
{\normalfont\normalsize\em}}
\makeatother

\newtheorem{proposition}{Proposition}

\newtheorem{lemma}{Lemma}
{\theoremstyle{definition}
\newtheorem{definition}{Definition}
\newtheorem{example}{Example}}
\theoremstyle{remark}
\newtheorem{remark}{Remark}

\begin{document}


\title[LRC codes with local error detection]{Locally Recoverable codes with local error detection}

\author{Carlos Munuera} 
\address{Department of Applied Mathematics, University of Valladolid, Avda Salamanca SN, 47014 Valladolid, Castilla, Spain}
\email{cmunuera@arq.uva.es}

\begin{abstract} 
A locally recoverable code is an error-correcting code  such that any erasure in a coordinate of a codeword can be recovered from a set of other few coordinates.  In this article we introduce a model of local recoverable codes that also includes local error detection.
The cases of the Reed-Solomon and Locally Recoverable Reed-Solomon codes are treated in some detail.   
\end{abstract}

\thanks{Research supported by  grant MTM2015-65764-C3-1-P MINECO/FEDER}
\maketitle


\section{Introduction}
\label{Intro}

Locally recoverable  codes  were introduced in \cite{GHSY}, motivated by the  use of coding theory techniques applied to distributed and cloud storage systems.
The growth of the amount of stored data  make the loss of information due to node failures a major problem.
To obtain a reliable storage,  when a node fails we want to recover the data it contains by using information from the other nodes. This is the {\em repair problem}. A method to solve it is to protect the data using error-correcting codes, \cite{GHSY}.
As typical examples of this solution, we can mention Google and Facebook, that use Reed-Solomon (RS) codes in their storage systems. The procedure is as follows: the information to be stored is a long sequence $b$ of symbols, which are elements of a finite field $\mathbb{F}_{\ell}$. This sequence is cut into blocks, $b=b_1,b_2,\dots$, of the same length,  $m$. According to the isomorphism $\mathbb{F}_{\ell}^m\cong \mathbb{F}_{\ell^m}$, each of these blocks can be seen as an element of the finite field $\mathbb{F}_q$, $q=\ell^m$. Fix an integer $k<q$. The vector $\mathbf{b}=(b_1,\dots,b_k)\in\mathbb{F}_q^k$ is encoded by using a RS code of dimension $k$ over $\mathbb{F}_q$, whose length $n$, $k<n \le q$, is equal to the number of nodes that will be used in its storage.  Then we choose $\alpha_1,\dots,\alpha_n\in\mathbb{F}_q$,  and send $b_1+b_2\alpha_i+\dots+b_k\alpha_i^{k-1}$ to the $i$-th node. When a node fails, we may recover the data it stores 
by using Lagrangian interpolation from the information of any other $k$ available nodes.

Of course other codes, apart from RS, can be used to deal with the repair problem.
Roughly speaking we can translate this problem in terms of coding theory as follows:
Let $\mathcal{C}$ be a linear code of length $n$ and dimension $k$ over $\mathbb{F}_{q}$.
A coordinate $i\in\{ 1,\dots,n\}$ is {\em locally recoverable with locality $r$} if there is a {\em recovery set} $R\subseteq \{1,\dots,n\}$ with $i\not\in R$ and $\# R= r$, such that for any codeword $\mathbf{x}\in\mathcal{C}$, an erasure in a coordinate $x_i$ of $\mathbf{x}$  can be  recovered by using the information given by the coordinates of $\mathbf{x}$ with indices in $R$. 
The code $\mathcal{C}$ is {\em locally recoverable (LRC) with locality} $\le r$   if each coordinate is so.  The {\em locality} of $\mathcal{C}$ is the smallest $r$ verifying this condition. For example, MDS codes of dimension $k$ (and RS codes in particular) have locality $k$. In Section \ref{LRC} we will specify some of these definitions.

Local recovery, understood in the previous terms, presents a clear drawback: when any of the coordinates used to recover $x_i$ contains an error, then the recovery will be wrong and both errors will remain undetected.  Moreover, the new errors created in this way could overcome the correcting capacity of $\mathcal{C}$ and therefore they should be impossible to eliminate, even using the code globally. Therefore, it may be appropriate to consider recovery sets that also allow local detection of errors in the coordinates used in the recovery process.

In this article we propose a model of locally recoverable codes that also gives local error detection. This is done in Section \ref{LREDC}. In Section \ref{Examples}  the particular cases of RS and LRC-RS codes will be dealt in  detail. 
As we will see, for RS codes it is enough to include one more coordinate in a recovery set to enable error detection.

\section{Locally recoverable codes}
\label{LRC}

In this section we state some definitions and facts concerning LRC codes that will be necessary for the rest of the work.
Let $\mathcal{C}$ be a $[n,k,d]$ code. Let $\mathbf{G}$ be a generator matrix of $\mathcal{C}$ with columns $\mathbf{c}_1,\dots, \mathbf{c}_n$. Given a set $R\subset \{1,\dots,n\}$ and a coordinate $i\notin R$, we say that $R$ is a {\em recovery set} for $i$ if $\mathbf{c}_i\in \langle \mathbf{c}_j : j\in R \rangle$, the linear space spanned by $\{ \mathbf{c}_j : j\in R\}$, see \cite{GHSY}.

Let $\pi_R:\mathbb{F}_q^n\rightarrow \mathbb{F}_q^r$ be the projection on the coordinates of $R$, where $r=\# R$. For $\mathbf{x}\in \mathbb{F}_q^n$, we write $\mathbf{x}_R=\pi_R(\mathbf{x})$. We will consider the punctured and shortened codes $\mathcal{C}[R]=\{ \mathbf{x}_R : \mathbf{x}\in \mathcal{C} \}$ and $\mathcal{C}[[R]]=\{ \mathbf{x}_R  : \mathbf{x}\in \mathcal{C}, \mbox{supp}(\mathbf{x})\subseteq R \}$. The following relation between these codes is well known \cite[Prop. 3.1.17]{Pel2}. Here we denote by $\mathcal{C}^{\perp}$ the dual of $\mathcal{C}$.

\begin{lemma} \label{dualL}
$\mathcal{C}[R]^{\perp}=\mathcal{C}^{\perp}[[R]]$.
\end{lemma}

Note that $\mathbf{c}_i\in \langle \mathbf{c}_j : j\in R \rangle$ if and only if $\dim(\mathcal{C}[R])=\dim(\mathcal{C}[\overline{R}])$, where $\overline{R}=R\cup\{ i\}$, so the notion of recovery set does not depend on the generator matrix chosen. In this case, there exist $w_1,\dots,w_n\in\mathbb{F}_q$ such that  $\sum w_j\mathbf{c}_j=0$ with $w_i\neq 0$ and $w_j=0$ if $j\notin \overline{R}$. Then  $\mathbf{w}=(w_1,\dots,w_n)\in \mathcal{C}^{\perp}$ and $\mathbf{w}_{\overline{R}}\in\mathcal{C}^{\perp}[[\overline{R}]]$.
So we have the following result.

\begin{lemma} \label{palabradual}
$R$ is a recovery set for a coordinate $i$ if and only if there exists a word $\mathbf{w}_{\overline{R}}\in\mathcal{C}^{\perp}[[\overline{R}]]$ with $w_i\neq 0$. In this case $\# R \ge d(\mathcal{C}^{\perp})-1$.
\end{lemma}

The smallest cardinality of a recovery set $R$ for coordinate $i$ is the {\em locality} of  $i$. The locality of $\mathcal{C}$ is the largest locality of any of its coordinates.

A word $\mathbf{w}_{\overline{R}}\in\mathcal{C}^{\perp}[[\overline{R}]]$ with $w_i\neq 0$  not only provides a recovery set but a recovery method: for every $\mathbf{x}\in \mathcal{C}$ we have $\mathbf{w}\cdot\mathbf{x}=\mathbf{w}_{\overline{R}}\cdot \mathbf{x}_{\overline{R}}=0$, where $\cdot$ denotes the usual inner product,  $\mathbf{w}\cdot\mathbf{x}=w_1x_1+\dots+w_nx_n$, so
\begin{equation}\label{rec}
x_i=x_i(\mathbf{w})=-w_i ^{-1}(\mathbf{w}_{R}\cdot \mathbf{x}_{R}).
\end{equation}
Then an undetected error in $\mathbf{x}_{R}$ leads to a wrong recovering of $x_i$.  
Such an error could be detected by using the detecting capability $\mathcal{C}$ or, alternatively, by using several recovery sets for coordinate $i$ (if available). But the first option goes against the local character of our method, while the second one increases the number of coordinates involved, worsening the probability of error, and does not allow us to determine which of the recovery sets used contains the error.

\section{Locally recoverable error-detecting codes}
\label{LREDC}

In this section we slightly modify the notion of recovery set to allow local detection of errors. Our starting point is the following result. 

\begin{lemma} \label{d-dim}
The minimum distance of $\mathcal{C}$ is $\ge d$ if and only if for all $S\subseteq \{1,\dots,n\}$ with $\# S> n-d$ we have $\dim(\mathcal{C}[S])=\dim(\mathcal{C})$.
\end{lemma}

The proof of this Lemma can be found in \cite[Prop. 4.3.12]{Pel2}. Next proposition is a direct consequence of  Lemma \ref{d-dim} and leads us to the following definition of recovery set detecting errors.

\begin{proposition}\label{d=2}
A set $\overline{R}\setminus\{i\}$ is a recovery set for every coordinate $i\in\overline{R}$ if and only if $d(\mathcal{C}[\overline{R}])>1$.
\end{proposition}

\begin{definition}\label{def}
A set $R\subseteq\{1,\dots,n\}$ is called a {\em recovery set detecting $t\ge 0$ errors} (or simply a {\em $t$-edr set} for short) for a coordinate $i\not\in R$ if $d(\mathcal{C}[\overline{R}])>t+1$, where $\overline{R}=R\cup\{i\}$.
\end{definition}

Then, a $t$-edr set $R$ for a coordinate $i$ is a recovery set for $i$; furthermore $R\cup\{i\}\setminus\{j\}$ is a $t$-edr set for all $j\in R$. If $\# R=r$ note that according to Lemma \ref{d-dim}, $R$ is a $t$-edr set for $i$ if and only if $\dim(\mathcal{C}[S])=\dim(\mathcal{C}[\overline{R}])$ for all $S\subseteq \overline{R}$ with $\# S\ge r-t$, where as above $\overline{R}=R\cup\{i\}$.
In particular $\dim(\mathcal{C}[\overline{R}])\le r-t$.
On the other hand, since $d(\mathcal{C}[\overline{R}])>t+1$ implies $d(\mathcal{C}[R])\ge t+1$, up to $t$ errors in any codeword $\mathbf{x}_R$, $\mathbf{x}\in\mathcal{C}$, may be detected, \cite[Sect. 2.4.1]{Pel2}. So when at most $t$ errors occur in $\mathbf{x}_R$, an  $t$-edr set $R$ either detects that errors occurred or either gives the correct value of $x_i$.  Checking for errors and recovering erasures may be performed by using appropriate words from the dual (as in (\ref{rec}), see examples of Section \ref{Examples}) or by other methods.

The minimum cardinality of a $t$-edr set for coordinate $i$ is the $t$-locality of $i$.  The code $\mathcal{C}$ is called {\em locally recoverable $t$-error-detecting code} ($t$-LREDC) if for every coordinate, a recovery set  detecting $t$ errors exists. Note that every code of minimum distance $d>t+1$ is a $t$-LREDC code (simply take $\overline{R}=\{1,\dots,n\}$). 
The maximum over the $t$-localities of all coordinates  is the $t$-locality of $\mathcal{C}$, denoted $r_t=r_t(\mathcal{C})$. For example,
since puncturing $<d$ times an MDS code gives a new MDS code of the same dimension, the $t$-locality of a $[n,k,d]$ MDS code with $d>t+1$ is $r_t=k+t$.  

Next we give two bounds on $r_t(\mathcal{C})$. The first one generalizes the bound given in Lemma \ref{palabradual} for $r=r_0$, by using generalized Hamming weights (see \cite[Sect. 4.5.1]{Pel2} for the definition of these weights).

\begin{proposition}\label{dualP}
Let $\mathcal{C}$ be a $[n,k,d]$ $t$-LREDC code. The $t$-locality  $r_t$ of $\mathcal{C}$ verifies
$r_t(\mathcal{C})\ge d_{t+1}(\mathcal{C}^{\perp})-1$, where $d_{t+1}(\mathcal{C}^{\perp})$ is the $(t+1)$-th generalized Hamming weight of $\mathcal{C}^{\perp}$.
\end{proposition}
\begin{proof}
Let $R$ be a $t$-edr set for a coordinate $i$. The condition $\dim(\mathcal{C}[\overline{R}])\le \# R-t$ implies $\dim(\mathcal{C}^{\perp}[[\overline{R}]])\ge t+1$, hence $\# R\ge d_{t+1}(\mathcal{C}^{\perp})-1$.
\end{proof}

The second bound on $r_t$ generalizes the well known Singleton-like bound, see \cite[Thm. 5]{GHSY}
\begin{equation} \label{LRCbound}
n+2\ge k+d+\left\lceil{\frac{k}{r_0}}\right\rceil .
\end{equation}

\begin{proposition}\label{t-LRCboundp}
Let $\mathcal{C}$ be a $[n,k,d]$ $t$-LREDC code. The $t$-locality of $\mathcal{C}$ verifies
\begin{equation}\label{t-LRCboundeq}
n+t+2\ge k+d+\left\lceil{\frac{k}{r_t-t}}\right\rceil (t+1).
\end{equation}
\end{proposition}

The proof of Proposition \ref{t-LRCboundp} is similar to that of (\ref{LRCbound}) given in \cite{GHSY}, taking into account the comments made after Definition \ref{def}, so we will omit it here. For similarity to the case $t=0$, we will say that the code $\mathcal{C}$ is $t${\em -optimal} if its $t$-locality reaches equality in (\ref{t-LRCboundeq}).

\section{Two examples of LREDC codes}
\label{Examples}

We present two examples of LREDC codes related to Reed-Solomon ones. We restrict to codes detecting one error, that is to $t=1$. Recall again that RS codes are  the most used in practice for recovery purposes.

We shall manage our examples with the language of evaluation codes. Let us remember that given a curve $\mathcal{X}$ over $\mathbb{F}_q$, a set of $n$ points $\mathcal{P}\subseteq\mathcal{X}(\mathbb{F}_q)$ and a function $f\in \mathbb{F}_q(\mathcal{X})$, we define the evaluation  of $f$ at $\mathcal{P}$ as $\mbox{ev}_{\mathcal{P}}(f)=(f(P))_{P\in\mathcal{P}}\in\mathbb{F}_q^n$. If $V\subseteq \mathbb{F}_q(\mathcal{X})$ is a linear space, then the set
$C(\mathcal{P},V)=\{ \mbox{ev}_{\mathcal{P}}(f) : f\in V\}$
is a linear code, called {\em evaluation code}, whose parameters can be studied by using resources from algebraic geometry, \cite{MO}. 
In the particular case $\mathcal{X}=\mathbb{A}(\mathbb{F}_q)$, the affine line over $\mathbb{F}_q$, and $V=\mathbb{F}_q[x]_{\le m}$, the set of polynomials with degree at most $m<\#\mathcal{P}$, then $C(\mathcal{P},V)$ is a Reed-Solomon code of dimension $k=m+1$, usually denoted $RS(\mathcal{P},m)$.  It is well known that $RS(\mathbb{A}(\mathbb{F}_q),m)^{\perp}=RS(\mathbb{A}(\mathbb{F}_q),q-m-2)$ and that the locality of RS codes is $r_0=k$.

\subsection{LREDC's from RS codes}

Let $\mathcal{P}\subseteq\mathbb{A}(\mathbb{F}_q)$. Consider the code $\mathcal{C}=RS(\mathcal{P},k-1)$ of length $n=\# \mathcal{P}$ and dimension $k\le n-2$. Let $r=k+1$ and take a set $\overline{R}$ of $r+1$ coordinates corresponding to  $\overline{\mathcal{R}}\subseteq\mathcal{P}$. Then $\mathcal{C}[\overline{R}]=RS(\overline{\mathcal{R}},k-1)$ is again a RS code, so  $d(\mathcal{C}[\overline{R}])=3$ and $R=\overline{R}\setminus\{i\}$ is a  recovery set  detecting one error for all $i\in \overline{R}$. Let us see how the recovering process, including error detection, is carried out. 
According to Lemma \ref{dualL} we have
$\mathcal{C}[\overline{R}]^{\perp}=RS(\mathbb{A}(\mathbb{F}_q),q-r)[[\overline{R}]]$ and $\mathcal{C}[R]^{\perp}=RS(\mathbb{A}(\mathbb{F}_q),q-r)[[R]]$, with  
$\dim(\mathcal{C}[\overline{R}]^{\perp})=2$, $\dim(\mathcal{C}[R]^{\perp})=1$.
Define
\begin{equation} \label{F}
F(x)= \prod_{\gamma\in \mathbb{A}(\mathbb{F}_q) \setminus \overline{\mathcal{R}}}(x-\gamma) \in \mathbb{F}_q[x]_{\le q-r-1}.
\end{equation}
Assume that coordinate $i$ corresponds to the element $\alpha_i\in\mathcal{P}$. Let $\mathcal{R}=\overline{\mathcal{R}}\setminus\{\alpha_i\}$ and 
\begin{equation} \label{zw}
\mathbf{z}_R=\mbox{ev}_{\mathcal{R}}((x-\alpha_i) F(x))  
\; \mbox{ , } \; 
\mathbf{w}_{\overline{R}}=\mbox{ev}_{\overline{\mathcal{R}}}( F(x)).
\end{equation}
Thus $\mathcal{C}^{\perp}[[R]]=\langle \mathbf{z}_R\rangle$, $\mathbf{w}_{\overline{R}}\in \mathcal{C}^{\perp}[[\overline{R}]]$ with $w_i\neq 0$.
Let $\mathbf{x}\in \mathcal{C}$ with an erasure in $x_i$ and at most one error in $\mathbf{x}_R$. The word $\mathbf{z}_R$ allows error detection in $\mathbf{x}_R$, and  $\mathbf{w}_{\overline{R}}$ allows recovering of $x_i$. 
That is, under the assumption that $\mathbf{x}$ contains at most one error in  $\mathbf{x}_R$,  $\mathcal{C}[R]^{\perp}=\langle \mathbf{z}_R\rangle$ implies that $\mathbf{z}_R\cdot\mathbf{x}_R\neq 0$ if $\mathbf{x}_R$ contains an error and $\mathbf{z}_R\cdot\mathbf{x}_R= 0$ if $\mathbf{x}_R$ is error-free. In this case, since $w_i\neq 0$, the erasure at $x_i$ can be recovered from $\mathbf{w}_{\overline{R}}$ by using the formula  (\ref{rec})
$x_i=x_i(\mathbf{w})=-w_i^{-1} (\mathbf{w}_R\cdot\mathbf{x}_R)$.
Note that our method does not require polynomial interpolation.

The $1$-locality of RS codes of dimension $k$ is $r_1=k+1$. We get equality in the bounds of Propositions \ref{dualP} and \ref{t-LRCboundp}: RS codes are 1-optimal. A similar reasoning proves that the same happens for general MDS codes.

\begin{remark} \label{ccomplex}
When $q$ is large with respect to $r$ (as desirable), the computation of $F(\alpha)$, for $\alpha\in \overline{\mathcal{R}}$ in equation (\ref{zw}), can be done more efficiently (from a computational point of view) by noting that
$\prod_{\lambda\in\mathbb{F}_q^*} \lambda=-1$. Let
$$
\phi(\alpha)=\prod_{\gamma\in \overline{\mathcal{R}}, \gamma\neq\alpha}(\alpha-\gamma).
$$
Computing $\phi(\alpha)$ requires $r$ multiplications, instead of the $q-r-1$ required by $F (\alpha)$.
Then $F(\alpha)=-\phi(\alpha)^{-1}$. In this way, the computational complexity of recovering with error detection is $O(r^2\log^3q)$. 
If for every $i$ we fix the corresponding $\overline{R}$, then we may pre compute and store $\mathbf{w}_{\overline{R}}$. From it we can   deduce $\mathbf{z}_{R}$ and complete the process of recovery with complexity $O(r\log^3q)$.
\end{remark}

\subsection{LREDC's from LRC-RS codes}

In \cite{TB} a variation of RS codes for recovering purposes was introduced, obtaining the so-called LRC-Reed-Solomon codes. Next we show how a slight modification of these codes can be used to detect errors in local recovering.  Let $p(x)$ be a polynomial of degree $r+1$ over $\mathbb{F}_q$. Consider the plane affine curve $\mathcal{X}$ of equation $y=p(x)$ and the map $y:\mathcal{X}\rightarrow \mathbb{F}_q$. For $\beta\in\mathbb{F}_q$ let $\mathcal{P}_\beta=y^{-1}(\beta)$ be the fibre of $\beta$ and let  $\overline{\mathcal{R}}_\beta=\{ \alpha : (\alpha,\beta)\in \mathcal{P}_\beta\}$,  $U=\{\beta : \# \mathcal{P}_\beta=r+1\}$, $u=\# U$. If $u>0$ then set $\mathcal{P}=\cup_{\beta\in U} \mathcal{P}_\beta$ and $n=u(r+1)=\# \mathcal{P}$.
Consider also the linear space of functions
$$
V=\bigoplus_{i=0}^{r-2} \langle 1,y,\dots,y^{l_i} \rangle x^i
$$
where $l_i$ are non negative integers such that $\delta=\max\{(r+1)l_i+i : i=0,\dots,r-2\}<n$. The evaluation  map $\mbox{ev}_{\mathcal{P}}: V\rightarrow \mathbb{F}_q^n$ is injective,  so the evaluation code $\mathcal{C}=\mathcal{C}(\mathcal{P},V)$ has dimension  $\dim(\mathcal{C})=\dim(V)=\sum (l_i+1)$. Besides, the Goppa bound implies $d(\mathcal{C})\ge n-\delta$, \cite[Sect. 4.2]{MO}.

Let us see how the recovering process, including error detection, is carried out. Let $\overline{R}$ be a set of $r+1$ coordinates corresponding to a set $\overline{\mathcal{R}}=\overline{\mathcal{R}}_\beta$, $\beta\in U$. For simplicity we identify the coordinate $j\in \overline{R}$ to the element $\alpha_j\in\overline{\mathcal{R}}$. Since $y$ is constant on $\overline{\mathcal{R}}$, $y=\beta$, then each function $f\in V$ acts  over $\mathcal{P}_\beta$ as a univariate polynomial $f(x,\beta)\in \mathbb{F}_q [x]_{\le r-2}$. Thus $\mathcal{C}[\overline{R}]$ is a Reed-Solomon code, 
$\mathcal{C}[\overline{R}]=RS(\overline{\mathcal{R}}\,r-2)=RS(\mathbb{A}(\mathbb{F}_q),r-2)[\overline{R}]$ of minimum distance 3 and we can proceed as in the case of RS codes. In fact $\mathcal{C}$ may be seen as a piecewise RS code. 

Let $i\in \overline{R}$, $R=\overline{R}\setminus \{i\}$ and $\mathcal{R}=\overline{\mathcal{R}} \setminus\{\alpha_i\}$. Let $F(x), \mathbf{z}_{R}, \mathbf{w}_{\overline{R}}$, be defined as in equations (\ref{F}) and (\ref{zw}).
Then $\langle \mathbf{z}_R\rangle=RS(\mathbb{A}(\mathbb{F}_q),q-r)[[R]]=\mathcal{C}^{\perp}[[R]]$ and  $\mathbf{w}_{\overline{R}}\in RS(\mathbb{A}(\mathbb{F}_q),q-r-1)[[R]]=\mathcal{C}^{\perp}[[\overline{R}]]$ with $w_i\neq 0$.
If $\mathbf{x}\in \mathcal{C}$ has an erasure in $x_i$ and at most one error in $\mathbf{x}_R$ then, as in the RS case,  we have $\mathbf{z}_R\cdot\mathbf{x}_R\neq 0$ if $\mathbf{x}_R$ contains an error and $\mathbf{z}_R\cdot\mathbf{x}_R= 0$ if $\mathbf{x}_R$ is error-free. In this case, since $w_i\neq 0$, the erasure $x_i$ can be recovered from $\mathbf{w}_{R}\in\mathcal{C}^{\perp}[[\overline{R}]]$ using the formula  (\ref{rec}).

If we take $l_0=\dots=l_{r-2}=u-1$, then $k=\dim(\mathcal{C})$ is as large as possible, $k=(r-1)u$, and $d(\mathcal{C})\ge 3$. Then we get equality in the bound (\ref{t-LRCboundeq}), so $r_1=r$ and the code $\mathcal{C}$ is 1-optimal. The comments on the computational complexity  made in Remark \ref{ccomplex} remain true in this case.

\begin{example}
Let $q=13$ and $p(x)=x^4$, so $r=3$. The fibres of $y$ correspond to the sets 
$\mathcal{R}_1=\{1,5,-5,-1\}$, 
$\mathcal{R}_3=\{2,3,-4,-2\}$ and 
$\mathcal{R}_9=\{4,6,-6,-4\}$. Let $\mathcal{C}$ be the evaluation code coming from these fibres and the  space of functions $V=\langle1,y,y^2 \rangle \oplus  \langle1,y,y^2 \rangle x$. Then $\mathcal{C}$ has length $n=12$, dimension $k=6$ and minimum distance $d\ge 3$. Let $\overline{R}$ be the set of coordinates corresponding to $\mathcal{R}_1$. 
Let $\mathbf{x}$ be a codeword, say $\mathbf{x}=\mbox{ev}_{\mathcal{P}}(1+xy)$,  with an erasure at the first coordinate, so that we have $\mathbf{x}_{\overline{R}}=(?,6,9,0)$ (and $?$ should be 2). Thus $R=\overline{R}\setminus\{1\}$ and we compute $\mathbf{z}_R=(8,-1,6)$, $\mathbf{w}_{\overline{R}}=(3,2,-2,-3)$. First we check $\mathbf{z}_{R}\cdot\mathbf{x}_{R}=0$, so $\mathbf{x}_R$ is accepted as error-free. Then we deduce $x_1=x_1(\mathbf{w})=2$. 
\end{example}

LRC-RS codes were extended to locally recoverable codes over arbitrary curves and rational maps in \cite{BTV,MT}. An analogous extension for  LREDC codes is straightforward.

\section{Conclusion}

In this article we have proposed a variation of local recoverable codes that offers the additional feature of local error detection, 
increasing the security of the recovery system.
The cases of the RS and LRC-RS codes have been specified in detail. In particular, we have shown that
for RS codes it is enough to include one more coordinate in a recovery set to enable error detection. \\


\noindent{\bf References}
\begin{thebibliography}{00}

\bibitem{BTV}
A. Barg, I. Tamo, S. Vladut, 
Locally recoverable codes on algebraic curves, 
IEEE Trans. Inform. Theory 63(8) (2017), 4928--4939.

\bibitem{GHSY}
P. Gopalan, C. Huang, H. Simitci and S. Yekhanin,
On the locality of codeword symbols, 
IEEE Trans.  Inform. Theory 58(11) (2012),  6925--6934.

\bibitem{MO}
C. Munuera, W. Olaya,
An introduction to algebraic geometry codes,
in Algebra for Secure and Reliable Communication Modelling, 
Contemporary Math. 642, AMS, Providence, 2015, 87-117. 

\bibitem{MT}
C. Munuera, W. Tenorio,
Locally recoverable codes from rational maps,
Finite Fields th App. 54 (2018), 80-100. 

\bibitem{Pel2}
R. Pellikaan, X.W. Wu, S. Bulygin, R. Jurrius,
Codes, Cryptology and Curves with Computer Algebra,
Cambridge University Press, Cambridge, 2017.

\bibitem{TB}
I. Tamo and A. Barg, 
A family of optimal locally recoverable codes,
IEEE Trans. Inform. Theory  60(8) (2014), 4661--4676.



\end{thebibliography}
\end{document}